\documentclass{cccg25}
\usepackage{graphicx,amssymb,amsmath}

\usepackage{import}
\usepackage{xcolor}
\usepackage{cleveref}
\newenvironment{moral}{\quote}{\endquote}
\newtheorem{definition}[theorem]{Definition}
\newcommand\RR{\mathbb{R}}
\DeclareMathOperator\im{im}
\usepackage{lineno}

\crefname{conj}{conjecture}{conjectures}
\crefname{obs}{observation}{observations}
\crefname{prop}{proposition}{propositions}
\crefname{lemma}{lemma}{lemmas}
\crefname{definition}{definition}{definitions}

\AddToHook{env/lemma/begin}{\crefalias{theorem}{lemma}}
\AddToHook{env/prop/begin}{\crefalias{theorem}{prop}}
\AddToHook{env/definition/begin}{\crefalias{theorem}{definition}}

\title{Square Packing with Asymptotically Smallest Waste Only Needs Good Squares}

\author{}
\author{Hong Duc Bui}

\begin{document}
\thispagestyle{empty}
\maketitle

\begin{abstract}
    We consider the problem of packing a large square with nonoverlapping unit squares.
    Let $W(x)$ be the minimum wasted area when a large square of side length $x$
    is packed with unit squares.
    In Roth and Vaughan's article that proves the lower bound $W(x) \notin o(x^{1/2})$,
    a good square is defined to be a square with inclination at most $10^{-10}$
    with respect to the large square.
    In this article, we prove that in calculating the asymptotic growth of the wasted space,
    it suffices to only consider packings with only good squares.
    This allows the lower bound proof in Roth and Vaughan's article
    to be simplified by not having to handle bad squares.
\end{abstract}

\section{Introduction}

There has been much research on the problem of packing unit squares into a large square.
Define $W(x)$ to be the minimum area of the wasted part when packing unit squares into
a large square of side length $x$.
On one hand, explicit constructions have been claimed in the literature
in \cite{Erdos_1975, Chung_2009, wang2016newresultpackingunit, Chung_2019},
which successively improves the upper bound to
$W(x) \in O(x^{0.637})$, $W(x) \in O(x^{0.631})$, $W(x) \in O(x^{0.625})$, $W(x) \in O(x^{0.6})$ respectively.
On the other hand, \cite{roth1978inefficiency} proves a lower bound
that $W(x) \notin o(x^{0.5})$.

In the proof of the lower bound $W(x) \notin o(x^{0.5})$ in \cite{roth1978inefficiency},
a good square was defined to be one with inclination
$\leq 10^{-10}$ with respect to the large square. The proof of \cite[lemma 4]{roth1978inefficiency}
uses the fact that there isn't too many bad squares.

In this article, we prove that it suffices to only consider the case where there are no bad squares,
simplifying the proof.  Furthermore, if some stronger lower bound is to be proven,
the proof can make use of this result to avoid having to handle bad squares.

The method used in this article can be generalized to the case of packing a large
quadrilateral with all angles almost equal to a right angle.

More formally, let $W^*(x)$ be the smallest wasted area when a large square of side length $x$ is packed with only good unit squares.
\begin{theorem}
    \label{theorem:small_tilt_asymptotic_packing}
    $W(x) \in \Theta(W^*(x))$.
\end{theorem}

The main technical contribution of this article is
an application of the max-flow min-cut theorem to this problem,
and a generalization of the fundamental lemma in \cite{roth1978inefficiency}
to handle squares far apart.

This article is organized as follows.
In \cref{sec_notation}, we formally define the square packing problem and the various tools we use throughout the article.
In \cref{sec_combinatorial_tools}, we prove a combinatorial tool that, given a set of marked cells on a grid,
bounds the total perimeter of rectangles that contains all these cells
in terms of 
the number of disjoint paths from these marked cells to the boundary on a grid.
In \cref{sec_geometric_tools}, we generalize the fundamental lemma in \cite{roth1978inefficiency}
to handle squares arbitrarily far apart.
In \cref{sec_main_results},
we describe a surgery procedure that cuts away all the bad cells in a packing,
and 
show that the additional wasted area is sufficiently small.
In \cref{sec_discussion},
we discuss some open problems that surface from our work.

\section{Notation}
\label{sec_notation}

\subsection{Notations Pertained to the Geometry of Square Packing}

We assume all squares live in an ambient Euclidean plane $\RR^2$ with a coordinate axis.

Let $S_0$ be an axis-aligned large square of side length $x$.

\begin{definition}[Packing]
    Let $S_0$ be fixed.
    A collection of unit squares $\mathcal A = \{ S_1, S_2, \dots, S_k \}$ is called a packing
    of $S_0$
    if the unit squares $S_i$ are non-overlapping and contained inside $S_0$.
    Define $|\mathcal A|$ to be the number of unit squares, so $|\mathcal A|=k$.
\end{definition}

From now on, $S_1, S_2, \dots$ will be used to denote unit squares in a packing.
The notation $S_i$ is the same as in \cite{roth1978inefficiency}.

\begin{definition}[$S_i$ as a set of points]
    When $S_i$ is used as a set of points on the plane, we implicitly mean the set of points on the boundary of the square.
    Here $i$ may either be $0$ ($S_0$ is the large square), or $i>0$.

    For $i>0$, the set of points on the boundary \emph{and inside} the square $S_i$ is denoted $S_i^*$.
    The set of points on the boundary \emph{and outside} the square $S_0$ is denoted $S_0^*$.
\end{definition}

With this definition, the intersection of $S_i^*$ and $S_j^*$ for any $0 \leq i < j \leq |\mathcal A|$
has area $0$.

\begin{definition}[Wasted area of a packing $W(\mathcal A)$]
    Let $W(\mathcal A) = x^2 - |\mathcal A|$, where as above, $x$ is the side length of $S_0$.
\end{definition}

Let us formally define the function $W(x)$ mentioned in the introduction.
\begin{definition}[The waste function $W(x)$]
    Let $\mathcal A$ be a packing with the largest number of unit squares
    for a fixed value of the side length $x$.
    Then define the minimum wasted area in packing a square of side length $x$ to be $W(x) = W(\mathcal A)$.

    Equivalently, $W(x) = \min_{\mathcal A} W(\mathcal A)$ over all packing $\mathcal A$ of a square $S_0$ with side length $x$.
\end{definition}

\begin{definition}[Wasted area in any shape]
    \label{def_wasted_area_arbitrary_shape}
    Fix the large square $S_0$ and a packing $\mathcal A$.
    Let $A$ be a measurable set on the plane.
    Then define $W(A) = |A \setminus (S_0^* \cup S_1^* \cup S_2^* \cup \cdots \cup S_k^*)|$.
    In words, $W(A)$ is the total area in $A$ that is inside the large square and not covered by any unit squares.
\end{definition}

\begin{definition}[Angle between two squares]
    Let $S_a$ and $S_b$ be squares
    (either an unit square used for packing, or the boundary).
    We say $S_a$ has the same orientation as $S_b$
    if some edge of $S_a$ is parallel to any edge of $S_b$.
    Let $\theta(S_a, S_b)$
    be the minimum rotation angle in radian of one of the squares
    to have the same orientation as the other.
    We allow either clockwise or counterclockwise rotation.
\end{definition}

Note that with our definition, $\theta(S_a, S_b) \leq \frac{\pi}{4}$ for every $a$ and $b$.

\begin{definition}[Distance between two objects]
    Let $p$, $q$ be points in the plane and $S$, $T$ be sets of points in the plane.
    Define $d(p, q)$ to be the distance from $p$ to $q$
    and $d(S, T) = \min_{s \in S, t \in T} d(s, t)$ whenever the minimum exists.
    Define $d(S, p) = d(S, \{ p \})$.
\end{definition}
When both $S$ and $T$ are closed, and at least one of them is compact, then $d(S, T)$ exists.
\begin{definition}[The open ball around an object]
    For a closed set $S$ in the plane, define
    \[ B(S, r) = \{\text{point }p\mid d(S, p) < r\} . \]
\end{definition}
\begin{definition}[Path on a plane]
    We define a path on the plane to be a continuous function
    $\gamma \colon[a, b] \to \RR^2$ for real numbers $a<b$.

    Define the image of the path on the plane to be
    $\im \gamma = \{ \gamma(x) \mid x \in [a, b] \} \subseteq \RR^2$.
\end{definition}

Following \cite{roth1978inefficiency}, we define the constant
$c = 10^{-10}$.

We also define a good square according to \cite{roth1978inefficiency}:
\begin{definition}[Good square and bad square]
    \label{def_good_square}
    A unit square $S_a \in \{ S_1, S_2, \dots \}$ is called good if $\theta(S_0, S_a) \leq c$.
    It is called bad if it is not good.
\end{definition}

\begin{definition}[Good packing, $W^*(x)$ function]
    \label{def_W_star}
    A packing $\mathcal A' = \{ S_1', \dots, S_k' \}$ is called good
    if all of the unit squares $S_1'$, $\dots$, $S_k'$ are good.
    Define $W^*(x) = \min_{\mathcal A'} W(\mathcal A')$
    over all good packings $\mathcal A'$ of the large square $S_0$ with side length $x$.
\end{definition}

The definition implies $W^*(x) \geq W(x)$.

\subsection{Notations Pertained to the Combinatorial Tools}

\begin{definition}[Grid, edge, neighborhood]
    Define a $n \times m$ grid to be a grid consisting of $n \times m$ squares (which we will call cells) glued side-by-side,
    such that there are $n$ rows and $m$ columns.  An \emph{edge} of the grid is defined to be an edge of any cell.

    The \emph{neighborhood} of a cell is defined to be the set of cells that shares an edge with that cell.
    This is also known in the literature as the $4$-neighborhood or von Neumann neighborhood of the cell.
    We also say the two cells to be \emph{adjacent}.

    A cell is said to be on the \emph{boundary} of the grid if it has less than four adjacent cells.
\end{definition}

See \cref{fig:example_grid} for an illustration.  We see that the grid has $(n+1) \times m$ horizontal edges,
$n \times(m+1)$ vertical edges,
and $n \cdot m - \max(0, n-2) \cdot \max(0, m-2)$ boundary cells.

\begin{definition}[Rectangle on a grid]
    Consider a $n \times m$ grid.
    A \emph{rectangle} on the grid is defined to be a tuple of integers $(i, i', j, j')$
    such that $1 \leq i \leq i' \leq n$, $1 \leq j \leq j' \leq n$.
    The cells of that rectangle is defined to be the set of all cells on row $i^*$ and column $j^*$
    for all integers $i \leq i^* \leq i'$ and $j \leq j^* \leq j'$.
    A cell is \emph{contained in} the rectangle if it belongs to the set of cells of the rectangle as above.
    The \emph{cellular perimeter} of the rectangle is defined to be $2 \cdot ((i'-i+1) + (j'-j+1))$;
    equivalently, this is the number of grid edges that is on the boundary of the rectangle.
\end{definition}
See \cref{fig:example_grid_rectangle} for an illustration.

We need to distinguish cellular perimeter and perimeter because later on there may be confusion if the side length of each cell is not exactly $1$ unit distance. Whenever clear, we will just say perimeter.

\begin{definition}[Bounding rectangle]
    For a non-empty set of cells $D$ in a grid,
    define its \emph{bounding rectangle} to be the rectangle with smallest perimeter
    and contains all the cells. It can be proven that this exists and is unique.
\end{definition}

\begin{figure}
    \centering
    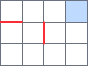
    \caption{Example of a $3 \times 4$ grid. Some edges are colored in red, and a boundary cell is colored in blue.}
    \label{fig:example_grid}
\end{figure}
\begin{figure}
    \centering
    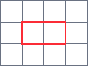
    \caption{Example of a rectangle on a grid. The rectangle is marked in red.}
    \label{fig:example_grid_rectangle}
\end{figure}

\begin{definition}[Path on a grid]
    Let $A$ be a cell in a grid.
    A \emph{path from $A$ to the boundary} is defined to be a sequence of cells
    $(P_1, P_2, \dots, P_k)$ such that $P_1 = A$, $P_k$ is on the boundary of the grid,
    all $P_i$ are different,
    and for each integer $1 \leq i < k$, $P_i$ and $P_{i+1}$ are adjacent.
\end{definition}

\begin{definition}[$8$-connectivity of a set of cells]
    Let $D$ be a set of cells on a grid.
    We say a cell is in the $8$-neighborhood of another (different) cell if
    they shares a vertex or an edge.
    We define a $8$-connected component of $D$ to be an equivalence class
    of the equivalence relation generated by the relation ``is in the $8$-neighborhood of''.
    $D$ is said to be $8$-connected if it only has one $8$-connected component.
\end{definition}

\section{Combinatorial Tools}
\label{sec_combinatorial_tools}

We need some helpers to prove the theorem.

\begin{lemma}
    \label{perimeter_bound}
    Let $D$ be a $8$-connected set of cells on a grid with $k>0$ cells.
    Let $R$ be its bounding rectangle.
    Then the perimeter of $R$ is no more than $4k$.
\end{lemma}
\begin{proof}
    Let $E$ be the edges of the grid that forms the boundary of $R$.
    Then the number of edges in $E$ is equal to the perimeter of $R$.

    Let $F$ be the edges of the cells that forms $D$, including duplicates
    by double-counting edges that belong to two cells in $D$.
    We have $|F|= 4k$.

    We will provide a injective mapping from $E$ to $F$ as follows.

    Consider the edges in $E$ that is on top of $R$.
    We gradually move each of them downwards until they hit a cell in $D$.
    Wherever the edge ends up at, it must belong to $F$.
    So we map the original edge in $E$ to this edge.
    See \cref{fig:eight_connected_boundary} for an illustration.

    \begin{figure}
        \centering
        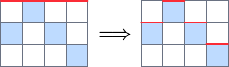
        \caption{Visualization for the mapping in the proof of \cref{perimeter_bound}.}
        \label{fig:eight_connected_boundary}
    \end{figure}
    Perform a similar procedure for the edges in $E$ that is on the other three sides of $R$.
    Because $D$ is $8$-connected, every row and every column has at least one cell in $D$,
    so each edge ends up somewhere instead of moving to infinity.

    It is then easy to see that the constructed mapping is injective.
    Therefore $|E|\leq|F|$, combining with $|F|=4k$ we're done.
\end{proof}

\begin{obs}
    \label{remark_merging_rectangles}
    An alternative way to prove \cref{perimeter_bound} is the following:
    we start with $k$ separate rectangles (each covering a cell in $D$) with perimeter $4$ each,
    then repeatedly merge the rectangles that have at least a point in common,
    noticing that the total perimeter is non-increasing during the process.
    Because of $8$-connectivity, the final result must be the bounding rectangle $R$.
    See \cref{fig:rectangle_merging} for an illustration.

    \begin{figure}
        \centering
        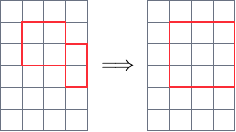
        \caption{Illustration of merging two rectangles that have at least a point in common.
        Here, they have a grid edge in common, and the total perimeter decreases by $2$ after merging.}
        \label{fig:rectangle_merging}
    \end{figure}
\end{obs}

\begin{prop}
    \label{claim:exist-rectangle-collection}
    Let $M$ be a collection of cells on a $n \times m$ grid, which we call the \emph{marked} cells.
    Then there exists an integer $f \geq 0$ such that:
    \begin{itemize}
        \item For each integer $1 \leq i \leq f$, there is a path $P_i = (P_{i, 1}, P_{i, 2}, \dots, P_{i, k_i})$ from a marked cell to the boundary (that is, $P_{i, 1} \in M$ is marked and $P_{i, k_i}$ is on the boundary);
        \item Each cell is used in at most one path. In other words, there is no $(i, j, i', j')$ such that $P_{i, j}= P_{i', j'}$
            but $(i, j) \neq(i', j')$.
        \item There is a collection of rectangles $\{ R_1, \dots, R_k \}$ of the grid such that all marked cells is contained in some rectangle,
            and their total perimeter is $\leq 4f$.
    \end{itemize}
\end{prop}

If we compare this with \cite{roth1978inefficiency},
the arguments in the proof of \cite[Lemma 4]{roth1978inefficiency}
can show that the cells in $M$ appears in either $\geq \sqrt{|M|}$ distinct rows
or $\geq \sqrt{|M|}$ distinct columns,
because otherwise there would be less than $\sqrt{|M|} \cdot \sqrt{|M|}$
locations which marked cells can be at, which is a contradiction.
This implies by taking the trivial vertical or horizontal paths
along the rows or columns containing a marked cell,
we get $f \geq \sqrt{|M|}$ distinct paths.
Our proof is stronger in that it can relate the value of $f$ with the total perimeter
of the bounding rectangles.

\begin{proof}
    We make use of the max-flow min-cut theorem. This is a very well-known theorem,
    one proof can be found in \cite{Ford_1956}.

    Construct the flow graph as follows.

    For each $1 \leq i \leq n$ and $1 \leq j \leq m$,
    construct node $A_{i, j}$ and $B_{i, j}$.
    Each cell on row $i$ and column $j$ corresponds to two nodes $A_{i, j}$ and $B_{i, j}$,
    which we will call the $A$-node and $B$-node of the cell respectively.
    Besides, there are source node $s$ and sink node $t$.

    For each cell, connect its $A$-node to its $B$-node with a (directed) edge with capacity $1$.

    For each cell on the boundary, connect its $B$-node to the sink $t$ with an edge with capacity $\infty$.

    For each marked cell, connect the source $s$ to the cell's $A$-node with an edge with capacity $\infty$.

    For each pair of adjacent cells $(i, j)$ and $(i', j')$,
    connect $B_{i, j}$ to $A_{i', j'}$ with an edge with capacity $\infty$.

    Note that when all edges with capacity $< \infty$ are removed, there is no path from $s$ to $t$,
    therefore the flow is finite.
    Since all edge weights are integers, the flow is an integer, let this be $f$.
    Furthermore, there exists a maximum flow where the amount of flow through each edge is an integer.

    If an algorithm such as Ford-Fulkerson is used to compute the flow, because all edge weights are integral,
    the flow value through each edge is also integral.
    Therefore, the flow can be decomposed into $f$ unit flows from $s$ to $t$, each has weight $1$.
    By taking the cells that corresponds to the edges with capacity $1$ along each path,
    we get a collection of paths $P_i = (P_{i, 1}, \dots, P_{i, k_i})$ for each $1 \leq i \leq f$.
    By construction, this is a path from $P_{i, 1}$ (which is a marked cell) to the boundary.

    Because each edge from a cell's $A$-node to its $B$-node has capacity $1$,
    each cell is used in at most one of the paths $P_i$, and each path only consist of distinct cells.

    From the procedure above, we have constructed a collection of paths $\{P_i\}$ as required.
    Now we need to construct the collection of rectangles $\{R_i \}$.

    By the max-flow min-cut theorem, there is a collection of edges with total capacity $f$
    such that cutting these edges results in no path from the source $s$ to the sink $t$.
    Because each edge with finite capacity has capacity $1$,
    there must be $f$ edges being cut; furthermore each of them connects some cell's $A$-node to its $B$-node.
    Let the set of cells corresponding to these cut edges be $C$.

    Because the edges corresponding to the cells in $C$ is a cut,
    for each marked cell, there is no path from that cell to the boundary that does not contain any cell in $C$.

    Divide $C$ into $8$-connected components, say $D_1$, $\dots$, $D_k$.
    For each $i$, let rectangle $R_i$ be the bounding rectangle of $D_i$.
    Applying \cref{perimeter_bound}, the total perimeter of $R_i$ is no more than $4 \sum_{i=1}^k |D_i| = 4f$.
\end{proof}

\begin{obs}
    This is known as the vertex splitting technique, which is also described
    in \cite{Ford_1963,Even_1975}.
    It is also well-known within competitive programming communities,
    one reference can be found in \cite[Section 8.4.5]{halim2018competitive}.

Constructed as is, the rectangles may have intersections.
The merging procedure described in \cref{remark_merging_rectangles} can be used to create
a collection of disjoint rectangles, but this is not needed in the places
\cref{claim:exist-rectangle-collection} is used.
\end{obs}

\section{Geometric Tools}
\label{sec_geometric_tools}

We make use of the fundamental lemma in \cite{roth1978inefficiency}, which we copy below.
\begin{lemma} \label{lemma:angle-diff}
	Let $S_a$ be a unit square used in the packing, and $S_b$ either the boundary of the shape being packed or another unit square.
    If the distance from $S_a$ to $S_b$ is at most $1$, then there is an open disk of radius $2$ containing $S_a$ such that the area of wasted space in the open disk is at least $c \cdot \theta(S_a, S_b)$.
\end{lemma}

In short:
\begin{moral}
	Difference in angle leads to wasted space.
\end{moral}

\Cref{lemma:angle-diff} only works for squares that are sufficiently close, however.
We want to generalize it as follows.

\begin{prop}
    \label{angle_diff_over_path}
    Let $S_a$ and $S_b$ be squares,
    $\gamma\colon[0, l] \to \RR^2$ be a path on the plane connecting $S_a$ and $S_b$.
	Then there is a constant $r>0$ and $c'>0$
    such that $W(B(\im \gamma, r)) \geq c' \cdot \theta(S_a, S_b)$.
\end{prop}
Recall $W(-)$ is the wasted area as defined in \cref{def_wasted_area_arbitrary_shape}.

Intuitively, we can imagine drawing a path $\gamma$ from $a$ to $b$,
then the statement says that if there is an angle difference $\theta(S_a, S_b)$ between $S_a$ and $S_b$,
the total wasted area inside the region colored green is at least proportional to $\theta(S_a, S_b)$.
See \cref{fig:intuit_waste_path} for an illustration.

\begin{figure}
    \centering
    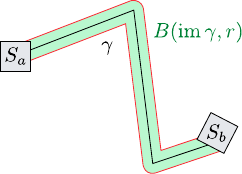
    \caption{Illustration for \cref{angle_diff_over_path}
    (drawing not to scale, $r=5$ in this article).}
    \label{fig:intuit_waste_path}
\end{figure}

\begin{obs}
    The vertical strip $T_i$ in \cite[Lemma 6]{roth1978inefficiency}
    or the vertical line segment $L(X)$ in \cite[Lemma 4]{roth1978inefficiency}
    is very similar to a vertical path
    $\gamma$ as we use here. Our proof can be seen as a generalization of the argument there.
\end{obs}


The main idea is the following.
\Cref{lemma:angle-diff} gives us \emph{circles} that contains large wasted area.
Then we find a sequence of squares $(S_{j_0} = S_a, S_{j_1}, S_{j_2}, \dots, S_{j_{k-1}}, S_{j_k} = S_b)$ along and near the path,
apply \cref{lemma:angle-diff} on each consecutive pair $(S_{j_i}, S_{j_{i+1}})$,
and get a collection of circles whose center is roughly along the path $\gamma$.
See \cref{fig:intuit_waste_path_proof} for an illustration.

\begin{figure}
    \centering
    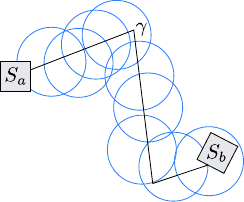
    \caption{Illustration for the construction of circles.}
    \label{fig:intuit_waste_path_proof}
\end{figure}

To prove the total wasted area in the union of the circles is sufficiently large,
we also need that each point is contained in $O(1)$ circles.

\begin{proof}
    If there exists any $0 \leq x \leq l$ such that $W(B(\gamma(x), \frac{1}{4})) \geq c \cdot \theta(S_a, S_b)$,
    we are done.

    Suppose otherwise. Construct a sequence of real numbers $\{ x_i \}$ as follows:
    let $x_1 = 0$, and for each integer $i \geq 2$, $x_i$ is defined to be the largest value of $x$
    such that $x_i \geq x_{i-1}$ and $d(\gamma(x_{i-1}), \gamma(x_i)) \leq \frac{1}{2}$.
    Because the domain of $\gamma$ is compact and $\gamma$ is continuous,
    this exists.

    Let $n$ be the smallest index such that $x_n = l$, if such index exists.
    Otherwise let $n = \infty$.

    For all $1 \leq j<k<n$,
    $d(x_j, x_k) \geq \frac{1}{2}$.
    Since $\im \gamma$ is compact, it is bounded.
    We see $n<\infty$
    by an application of the pigeonhole principle:
    say $\im \gamma$ is inside a square of side length $R$,
    divide it into $\lceil 4R \rceil^2$ small squares of side length $\leq \frac{1}{4}$,
    then no two points $(x_j, x_k)$ for $1 \leq j<k<n$ can be in the same square,
    therefore $n \leq \lceil 4R \rceil^2+1$.

    For each $1 \leq j \leq n$,
    define index $i_j \geq 0$ such that
    $d(\gamma(x_j), S_{i_j}^*) < \frac{1}{4}$.
    If there are multiple such indices, pick any;
    with the exception of $i_1 = a$ and $i_n = b$.

    We show such $i_j$ exists:
    Since we assume for all $0 \leq x \leq l$ then $W(B(\gamma(x), \frac{1}{4})) < c \cdot \theta(S_a, S_b)
    < \frac{1}{4}^2 \pi$,
    there is at least one square with nonzero intersection area with $B(\gamma(x), \frac{1}{4})$.

    Next, for each $1 \leq j<n$,
    \begin{align*}
    d(S_{i_j}^*, S_{i_{j+1}}^*)
    &\leq d(S_{i_j}^*, \gamma(x_j))
    + d(\gamma(x_j), \gamma(x_{j+1}))
    \\ & \qquad + d(\gamma(x_{j+1}), S_{i_{j+1}}^*)
    \\ &< \frac{1}{4} + \frac{1}{2} + \frac{1}{4} 
    \\ &= 1.
    \end{align*}
    Therefore $d(S_{i_j}, S_{i_{j+1}}) < 1$.

    If $\theta(S_{i_j}, S_{i_{j+1}}) = 0$
    then $W(B(\gamma(x_j), 5)) \geq 0 = c \cdot \theta(S_{i_j}, S_{i_{j+1}})$.
    Otherwise, at least one of
    $S_{i_j}$ or $S_{i_{j+1}}$ is an unit square,
    apply \cref{lemma:angle-diff},
    there is an open disk with radius $2$ containing either $S_{i_j}$ or $S_{i_{j+1}}$
    and have wasted area $\geq c \cdot \theta(S_{i_j}, S_{i_{j+1}})$.
    Let $p$ be the center of this disk, so the disk is $B(p, 2)$, and let $k \in \{ i_j, i_{j+1} \}$ be the index
    such that $S_k \subseteq B(p, 2)$.
    Then
    $d(\gamma(x_j), S_k) < \frac{3}{4}$,
    so $d(\gamma(x_j), p) \leq 2+\frac{3}{4}$,
    so $B(p, 2) \subseteq B(\gamma(x_j), 5)$,
    therefore $W(B(\gamma(x_j), 5)) \geq c \cdot \theta(S_{i_j}, S_{i_{j+1}})$.

    Taking the sum, we get
    \begin{align*}
        \sum_{j=1}^{n-1} W(B(\gamma(x_j), 5))
        & \geq c \cdot \sum_{j=1}^{n-1} \theta(S_{i_j}, S_{i_{j+1}}) \\
        & \geq c \cdot \theta(S_{i_1}, S_{i_n})
        = c \cdot \theta(S_a, S_b) .
    \end{align*}

    Since for each $1 \leq j<k<n$,
    $d(x_j, x_k) \geq \frac{1}{2}$,
    for each point $p$ on the plane,
    there can be at most $1600$ such points $x_j$
    in $B(p, 5)$,
    therefore
    \[
        W\Big(\bigcup_{j=1}^{n-1} B(\gamma(x_j), 5)\Big)
        \geq \frac{c}{1600} \cdot \theta(S_a, S_b) .
    \]

    This gives the desired conclusion.
\end{proof}

\section{Main Results}
\label{sec_main_results}

Now we return to the square packing problem.
Recall $x$ is the side length of the large square $S_0$.

\begin{prop}
    \label{reduce-to-good-packing}
    There exists a constant $c_2 > 1$ independent of $x$ such that:
    for all packing $\mathcal A$, there exists a good packing $\mathcal A'$
    such that $W(\mathcal A') \leq c_2 W(\mathcal A)$.
\end{prop}

\begin{proof}
    If $x \leq 1$, the statement is obvious. We assume $x>1$.

    Divide the large square into $\lceil x \rceil$ rows and that many columns,
    we get a $\lceil x \rceil \times \lceil x \rceil$ grid,
    each cell is a square with side length $>\frac{1}{2}$ and $\leq 1$.

    Let the set of marked cells $M$ be the set of cells that has any (positive area) overlap with a bad square.
    Apply \cref{claim:exist-rectangle-collection} on this set $M$,
    we get an integer $f \geq 0$, and a collection of paths $\{ P_i \}_{1 \leq i \leq f}$,
    where $P_i = (P_{i, 1}, \dots, P_{i, k_i})$.
    Let $O_{i, j}$ be the center of cell $P_{i, j}$,
    and let path $\gamma_i$ be the polyline consisting of the shortest segment from
    any bad square that has an overlap with cell $P_{i, 1}$
    to $O_{i, 1}$, followed by the polyline $O_{i, 1} O_{i, 2} O_{i, 3} \dots O_{i, k_i}$,
    followed by the shortest segment from $O_{i, k_i}$ to the boundary of $S_0$.

Apply \cref{angle_diff_over_path} on each of the path $\gamma_i$,
we get (where $c'$ and $r$ are as in \cref{angle_diff_over_path})
\[
    \sum_{i=1}^f W(B(\im \gamma_i, r)) \geq c' \cdot c \cdot f.
\]
Because each cell $P_{i, j}$ is only used once and the side length of each cell is between $\frac{1}{2}$ and $1$,
each point in $\bigcup_{i=1}^f B(\im \gamma_i, r)$ is covered by $\Theta(1)$ of the $B(\im \gamma_i, r)$,
therefore
\begin{align*}
    W(\mathcal A)
        & \geq W\Big(\bigcup_{i=1}^f B(\im \gamma_i, r) \Big) \\
        & \geq \frac{1}{\Theta(1)} \sum_{i=1}^f W(B(\im \gamma_i, r)) \\
        & \geq \frac{c' \cdot c \cdot f}{\Theta(1)} .
\end{align*}

    From the application of \cref{claim:exist-rectangle-collection} above,
    we also get a collection of rectangles $\{R_1, \dots, R_k\}$ that contains all the marked cells.
    Use these rectangles, modify $\mathcal A$ into $\mathcal A'$ as follows.
    First, delete all unit squares completely contained in $\bigcup_i R_i$,
    this way all bad squares are deleted.
    Then, pack as many axis-aligned unit squares with integral vertex coordinates as possible.

    Note that this packing procedure can only create wasted space along the perimeter of $R_i$.

    Because the side length of each cell is between $\frac{1}{2}$ and $1$,
    for each rectangle $R_i$, its perimeter is between $\frac{1}{2}$ and $1$ times its cellular perimeter.
    We know from \cref{claim:exist-rectangle-collection} that the total cellular perimeter of $R_i$ is $\in O(f)$,
    so the total perimeter of $R_i$ is also $\in O(f)$.

    Therefore the additional wasted area introduced by the modification procedure is $\in O(f)$,
    so $W(\mathcal A') \leq W(\mathcal A) + O(f)$.
    Since $W(\mathcal A) \in \Omega(f)$,
    we get the desired result.
\end{proof}

Now we can prove \cref{theorem:small_tilt_asymptotic_packing}.

\begin{proof}
    As mentioned after \cref{def_W_star}, $W(x) \leq W^*(x)$.
\Cref{reduce-to-good-packing} implies $W^*(x) \leq c_2 W(x)$.
\end{proof}

\section{Discussion}
\label{sec_discussion}

\subsection{Generalization of the Wasted Area Along Path Proposition}

How small can $r$ be in \cref{angle_diff_over_path}?
We guess that $r$ can be made arbitrarily small.

\begin{conj}
    Let $S_a$ and $S_b$ be squares.
    Let $\gamma\colon[0, l] \to \RR^2$ be a path on the plane connecting $S_a$ and $S_b$.
	Then there is a constant $c'>0$ such that
    for all $r>0$,
    $W(B(\im \gamma, r)) \geq c' \cdot \min(1, r^2) \cdot \theta(S_a, S_b)$.
\end{conj}

To prove this would require proving a stronger version of \cref{lemma:angle-diff}, however.

\subsection{Strengthening of the Result}

We have shown that only good squares (i.e. squares with tilt bounded by a constant) need to be considered
to compute the asymptotic growth of the wasted area.
Nonetheless, the tilt in the proposed packing methods \cite{Erdos_1975, Chung_2009, wang2016newresultpackingunit}
is in fact much smaller, $O(x^{-\epsilon})$ for some $\epsilon>0$.
Can the bound on the tilt be sharper,
for example $O(1/\log x)$ or even $O(x^{-\epsilon})$ for some $\epsilon>0$?

\subsection{Other Algorithmic Considerations}

The construction in \cref{claim:exist-rectangle-collection} can be explicitly implemented.
In doing so, notice that the amount of flow through each edge is either $0$ or $1$,
therefore the maximum flow can also be found if the edges with capacity $\infty$
are set to have capacity $1$ instead.
If the Dinic algorithm is used to compute the flow, the constructed flow graph has $O(n m)$ vertices and $O(n m)$ edges,
therefore the time complexity is $O((n m)^{3/2})$ according to \cite[Theorem 1]{Even_1975}.
Alternatively, because the flow value is $O(n+m)$, the simpler Edmonds--Karp algorithm
\cite{Edmonds_1972} or any implementation of the Ford--Fulkerson method \cite{Ford_1956}
would have time complexity $O(n m (n+m))$.

A simpler algorithm to find disjoint paths in \cref{claim:exist-rectangle-collection}
is to perform multi-source BFS from marked cells to the boundary, greedily pick the shortest
path found each time, while avoiding cells that already belong to some path.
Let $f'$ be the number of disjoint paths found by this method, then $f' \leq f$.
While this method does not provide any guarantee of optimality
(it may happen that
there is no collection of rectangles with total perimeter $\leq 4f'$ covering all the marked cells,
see \cref{sec_test_kill_multisource_bfs}),
numerical experiments suggests the conjecture $f' \in \Omega(f)$.
We find this likely because of the special structure of the grid.

In implementing the merging procedure described in \cref{remark_merging_rectangles},
because the perimeter of each rectangle is at most $2(n+m)$,
and the perimeter increases by at least $2$ after a merge step with another rectangle
that is not completely contained inside,
the merging procedure terminates after at most $n+m$ iterations.
Each iteration can be implemented in $O(n m)$, therefore the total time complexity is
$O(n m (n+m))$.

It would be interesting to investigate whether there is any algorithm taking less than cubic
time for each of the problems above when $n \in \Omega(m)$.

\section{Conclusion}
\label{sec_conclusion}

We prove a relation between arbitrary packing and good packing of a large square,
therefore anyone proving a lower bound on $W(x)$
in the future similar to \cite{roth1978inefficiency}
only need to consider good packings.

\section*{Acknowledgements}

The author would like to thank some friends for insightful discussion and finding some typos.

\clearpage



{
\small
\bibliographystyle{abbrv}

}

\section*{Appendix}
\subsection{Example Grid where Greedy Multi-source BFS Gives Suboptimal Result}
\label{sec_test_kill_multisource_bfs}

Consider the grid depicted in \cref{fig:test_kill_multisource_bfs},
where the marked cells are colored green or blue.
Multi-source BFS algorithm would first greedily find the paths to the blue cells,
then it would not be able to find any additional paths to the green cells.
However, if the $8$ paths marked red are removed to
create paths to the green cells first, more cells would be created.
It can be seen that in this case greedy BFS algorithm finds $f' = 24$ disjoint paths, but
there is no collection of rectangles with total perimeter $4f' = 96$ covering all the marked cells.

\begin{figure}
    \centering
    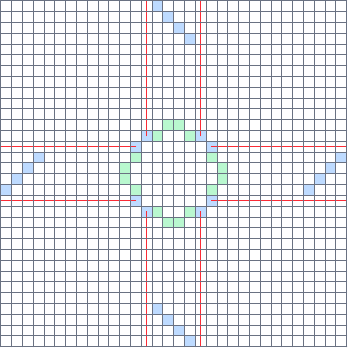
    \caption{Example grid where greedy multi-source BFS algorithm gives $f'$ disjoint paths
    and $4f'$ is less than the minimum total perimeter of bounding rectangles.}
    \label{fig:test_kill_multisource_bfs}
\end{figure}

\end{document}